\documentclass[%
 reprint,
 amsmath,amssymb,
 aps, pra, floatfix
]{revtex4-2}
\pdfoutput=1

\usepackage{amsfonts}
\usepackage{amsmath}
\usepackage{amssymb}
\usepackage{amsthm}
\usepackage{stmaryrd}
\usepackage{physics}
\usepackage{braket}
\usepackage{dsfont}
\usepackage{bbold}
\usepackage{graphicx}
\usepackage{relsize}
\usepackage{centernot}
\usepackage{makecell}
\usepackage[table,xcdraw]{xcolor}
\usepackage{enumerate}
\usepackage[pagebackref, colorlinks = true, linkcolor = blue, urlcolor  = blue, citecolor = blue!90!black]{hyperref}
\usepackage[margin=1in]{geometry}
\usepackage{enumitem}
\usepackage{mathtools}
\usepackage{comment}
\usepackage{float}
\usepackage[capitalize]{cleveref}
\usepackage{booktabs}
\usepackage{multirow}
\usepackage[skins]{tcolorbox}
\usepackage{nicematrix}
\NiceMatrixOptions{cell-space-limits = 2pt}
\usetikzlibrary{matrix, arrows.meta}
\usepackage{fancyhdr}
\usepackage{etoolbox}

\usepackage{fancyhdr}
\usepackage{adjustbox}
\usepackage{standalone}

\usetikzlibrary{graphs, graphs.standard}%

\usepackage{mathtools}
\newcommand{\restr}[2]{\raisebox{-.5ex}{$\left.#1\right|_{#2}$}}

\newcommand{\PPTcone}{\mathsf{PPT}}

\newcommand{\SEP}[1]{\mathsf{SEP}_{#1}}
\newcommand{\SEPcone}{\mathsf{SEP}}

\renewcommand{\epsilon}{\varepsilon}
\renewcommand{\phi}{\varphi}
\newcommand{\overbar}[1]{\mkern 1.5mu\overline{\mkern-1.5mu#1\mkern-1.5mu}\mkern 1.5mu}

\newcommand{\MLDUI}[1]{\mathcal{M}_{#1}
\newcommand{\rom}[1]{\expandafter\@slowromancap\romannumeral #1@}
(\mathbb{C})^{\times 2}_{\mathbb{C}^{#1}}}

\newcommand{\C}{\mathbb{C}}
\renewcommand{\C}[1]{\mathbb{C}^{#1}}
\newcommand{\M}[1]{\mathcal{M}_{#1}(\mathbb{C})}

\newcommand{\id}[1]{\mathrm{I}_{n}}

\newcommand{\channel}{\mathrm{\Phi}}

\theoremstyle{definition}
\newtheorem{theorem}{Theorem}[section]
\newtheorem{definition}[theorem]{Definition}

\newtheorem{proposition}[theorem]{Proposition}
\newtheorem{corollary}[theorem]{Corollary}
\newtheorem{lemma}[theorem]{Lemma}
\newtheorem{remark}[theorem]{Remark}

\newtheorem{example}[theorem]{Example} 
\begin{document}

\title{Rank-based entanglement detection for bound entangled states}

\author{Aabhas Gulati}
\email{aabhas.gulati AT math.univ-toulouse.fr}
\affiliation{Institut de Mathématiques de Toulouse, France}

\begin{abstract}
    We develop a new method for entanglement detection in bipartite quantum states by using the violation of the rank-1-generated property of matrices. The positive-semidefinite matrices form a convex cone that has extremal elements of rank 1. But, convex conic subsets resulting from the presence of linear constraints allow extremal elements of rank $\geq 2$. The problem of deciding when a matrix is rank-1 generated, i.e, a sum of rank-1 positive-semidefinite matrices, has been studied extensively in optimization theory. This rank-1 generated property acts as an entanglement criterion, and we use this property to find novel classes of PPT (Positive under partial transposition) entangled states. We do this by mapping some faces of PPT density matrices to convex cones that are not rank-1 generated. We show that all separable states map to a rank-1 generated matrix. In general, the same is not true for the corresponding mapping of PPT entangled states. We also extend this approach to construct PPT entangled edge states. Finally, we provide witnesses that detect violations of the rank-1 generated property.
\end{abstract}

\maketitle
\section{Introduction}
With the advent of quantum information, entanglement is now recognized as an essential resource, enabling tasks such as quantum cryptography \cite{Ekert1991crypto}, quantum teleportation \cite{Bennett1993teleport}, and measurement-based quantum computation \cite{Raussendorf2001computer}. Despite significant progress in this field, the theoretical understanding of entanglement in mixed quantum states remains challenging. Even the fundamental question of deciding if any given mixed quantum state is entangled or separable is NP-hard \cite{gurvits2003classical,gharibian2008strong}. This fact alone points to a fundamental hardness in detecting entanglement in arbitrary states, implying that a complete characterization in terms of necessary and sufficient conditions that can be checked efficiently is unattainable (of course, unless P=NP). Still, several algorithms whose complexities scale quasi-polynomially in the system size have been developed \cite{brandao2011quasipolynomial}. The complexity of this problem led to the strategy of necessary separability criteria, such that any state that violates any of these criteria is certainly entangled. Many such criteria to detect entanglement in mixed quantum states have been identified; for a comprehensive review, see \cite{guhne2009entanglement}. 

The PPT (Positive under partial transposition) criterion, \cite{peres1996separability}, sometimes also referred to as the Peres-Horodecki criterion, is arguably the most popular criterion to detect entanglement. This test can be applied to quantum states by checking if the matrix obtained by partial transposition on one subsystem is positive semidefinite. If not, the state is entangled. Surprisingly, this test is also sufficient in two-qubit and qutrit-qubit systems \cite{Strmer1963,horodecki1996separability}. In quantum systems of other dimensions, there exist states that are PPT but are still entangled. The first few examples of two qutrit and a qubit-ququart were constructed in \cite{Horodecki1997PPTent}. The PPT criterion is also sufficient for the class of quantum systems with complete unitary symmetries, the Werner states, the isotropic states \cite{vollbrecht2001entanglement}, and their generalizations with weaker symmetries \cite{park2024universal}. 
Although the PPT criterion originally had purely mathematical motivations, the states that are PPT entangled were found to satisfy quite interesting physical properties. Particularly, these states cannot be converted into pure entangled states under LOCC operations \cite{Horodecki1998bound}, a fact that is intimately connected to the fundamental irreversibility of entanglement \cite{Vidal2001irreversible}. For this reason, they are sometimes called bound entangled states. Although for high enough dimensions, random-matrix arguments imply that these states have significant volume for some chosen noise thresholds \cite{aubrun2014entanglement,louvet2025bound,nechita2025random}, the systemic construction of examples of PPT entangled states remains open, and highly researched \cite{breuer2006optimal, spengler2012entanglement,piani2007class,singh2020entanglement, ghimire2023quantum}. Recently, some bound entangled states with high Schmidt numbers were also constructed in \cite{pal2019class, huber2018high, krebs2024high}. 

In this article, we present a new technique in entanglement detection, which we use to find novel examples of PPT entangled quantum states. We do this by making useful connections between PPT entanglement and the so-called ``rank-1 generated property" for conic subsets of positive-semidefinite (PSD) cones. This area of research has been prominent in matrix analysis and optimization theory, see, for example, \cite{argue2023necessary,hildebrand2016spectrahedral}. It is well-known that the convex cone of PSD matrices has extremal elements of rank-1, and any element in the PSD cone is a finite sum of such rank-1 elements. In general, this is not true for conic subsets of PSD, i.e, by imposing appropriate constraints on the matrices, the ranks of some of the extremal elements of the corresponding conic subset can be strictly greater than $1$. Many examples of such conic subsets can be found in \cite{argue2023necessary,hildebrand2016spectrahedral}. In general, this problem is usually understood on a case-by-case basis \cite{argue2023necessary}, and only a few general conditions for a cone to satisfy this property are known, only when the number of constraints is small.

Let us briefly summarise our approach here: We map the bipartite PPT states to a constrained PSD convex cone that is not rank-1 generated. The main observation is that whenever a state is separable, this mapping results in a rank-1 generated matrix, i.e, a sum of \emph{only} rank-1 elements of a constrained PSD cone. As we show, this provides a non-trivial necessary condition for separability, since any matrix that is mapped to any element that is not rank-1 generated can only arise from a PPT entangled state. To test the validity of the considered approach, we provide constructions of such a mapping to convex cones where the rank-1 generated property has already been studied in the literature, and for different faces of the PSD cone, we can detect PPT entangled states. Although the core idea of the method is quite simple to state, we achieve a wide variety of constructions within this approach, with a high number of free parameters. In effect, this approach is analogous to the range criterion for entanglement, which detects PPT entangled states in various cases, but is not as easy to compute for arbitrary states.

We now provide the outline of this article. In \cref{sec:prelims}, we discuss the preliminary topics to understand this article. In particular, we focus on separability of quantum states, local diagonal orthogonally invariant (LDOI) states, the rank-1-generated property of convex cones, and some graph-theoretic notions. In \cref{sec:main-results}, we discuss the core ideas, along with the main result in \cref{thm:rank-1-generated-ent}, that provides a method to detect entanglement by checking whether an element in a convex conic subset of PSD is rank-1 generated. In \cref{sec:entanglement-detection}, we show that the method can be applied to successfully detect entanglement in various families of states, and construct PPT entangled states in high dimensions and with a large number of parameters. These families of states we construct are connected to conic subsets of PSD for which the rank-1-generated property is well-studied. 
\section{Preliminaries}
\label{sec:prelims}
\subsection{Notation}
The sets of complex, real, and non-negative real numbers will be denoted as $\mathbb{C}, \mathbb{R}$ and $\mathbb{R}_+$ respectively. Let $[d]$ denote the finite set $\{1, 2 \ldots d\}. $ We adopt Dirac's \textit{bra-ket} notation to denote vectors, elements of $\mathbb{F}^d$ where $\mathbb{F} \in \{\mathbb{R}, \mathbb{C} \}$ and, the notation $\mathbb{R}_+^d$ will denote the set of $d$ dimesional non-negative vectors. We denote the standard basis vector $e_i\in \mathbb{C}^d$ with the ket \(\ket{i}\), and its dual vector, \(e_i^*\) as \(\bra{i}\). For vectors \(v, w \in \mathbb{C}^d\), the linear map \(vw^* : \mathbb{C}^d \rightarrow \mathbb{C}^d\) is written as \(\ketbra{v}{w}\), and it represents a rank-$1$ matrix. To denote complex, and real matrices of size $d \times d'$, we use the notation $\mathcal{M}_{d,d'}$ and, $\mathcal{M}_{d,d'}(\mathbb{R})$, respectively. If the dimensions $d = d'$, we shorten it to \(\mathcal{M}_d := \mathcal{M}_{d,d}\) and to \(\mathcal{M}_d (\mathbb{R}) := \mathcal{M}_{d,d} (\mathbb{R})\). The rank and range of a matrix will be denoted by $\operatorname{rk}(A)$ and $\operatorname{ran}(A)$, respectively. The set of self-adjoint matrices is denoted by $\mathcal{M}_d^{\operatorname{sa}}$. The set of complex (resp. real) positive semidefinite matrices  is denoted as $\mathcal{M}_d^+$ (resp. $\mathcal{M}^+_d(\mathbb{R})$)  and the set of matrices with non-negative entries as $\mathcal{M}_d(\mathbb{R}_+)$.  For bipartite positive semidefinite matrices, we use the notation, $(\mathcal{M}_d \otimes  \mathcal{M}_d)^+.$ We use $\operatorname{supp}(v)$ to denote the support of the vector $v \in \mathbb{C}^d$, i.e, $\operatorname{supp}(v) := \{i \in [d] \mid v_i \neq 0\}.$ We use $\odot$ to denote the entrywise product of vectors $(v \odot w)_i := v_i \cdot w_i$ and matrices $(A \odot B)_{ij} := A_{ij} \cdot B_{ij}.$ We use $\mathbb{I}_n$, and $\mathbb{J}_n$ to denote $n \times n$ identity matrix, and all ones matrix, respectively.

\subsection{Convex analysis}
In this section, we provide a brief introduction to convex analysis. As is usual in mathematical literature on entanglement theory, we work with convex cones rather than convex sets. Consider the following definition. 
\begin{definition}[Convex cones]
    Let $V$ be a real vector space. A \emph{convex cone} $\mathcal C$ is a non-empty subset of $V$ having the following two properties: 
    \begin{itemize}
        \item If $x \in \mathcal C$ and $\lambda \in \mathbb R_+ = [0, \infty)$, then $\lambda x \in \mathcal C$.
        \item If $x,y \in \mathcal C$, then $x+y \in \mathcal C$.
    \end{itemize}
    In particular, $0 \in \mathcal C$. 
\end{definition}    

The cone $\mathcal C$ is said to be \emph{pointed} if $\mathcal C \cap (-\mathcal C) = \{0\}$; in other words, $\mathcal C$ is pointed if it does not contain any line. The important elements of a convex cone are those that cannot be decomposed into other ones, called \emph{extremal} rays.  For a vector $v \neq 0$, the half-line $\mathbb R_+ v := \{\lambda v : \lambda\in \mathbb{R}_+ \} \subseteq \mathcal C$ is called an \emph{extremal ray} of $\mathcal C$ (we write $\mathbb R_+ v \in \operatorname{ext} \mathcal C$) if
$$v = x+y \quad\text{with } x,y \in \mathcal C \implies x,y \in \mathbb R_+ v.$$

The set of non-negative matrices, $\mathcal{M}_d(\mathbb{R}_+)$, and the set of positive semidefinite matrices, $\mathcal{M}_d^+$, are two well-studied convex cones in the literature. The extreme rays of both these cones can be completely characterized, as   
$\operatorname{ext}\mathcal{M}_d^+ = \{\mathbb{R}_+{\ketbra{v}{v}} \, | \, v \in \mathbb{C}^d\}$ and 
$\operatorname{ext}\mathcal{M}_d(\mathbb{R}_+) = \{\mathbb{R}_+{\ketbra{i}{j}} \}_{i,j}$ respectively. Moreover, the extremal rays of both these cones consist only of matrices that are rank-1.

In this article, $\operatorname{cone} (X)$ denotes the \emph{smallest} convex cone containing $X$. We call this the cone generated by the set $X$. In finite-dimensional vector spaces, a closed convex and pointed cone is generated by its extremal rays \cite{rockafellar1970convex}, but the same result is not necessarily true for infinite dimensions. 
\begin{definition}
Let $\mathcal{C}$ be a convex cone. Then, the convex sub-cone $\mathcal{F} \subseteq \mathcal{C}$ is called a \emph{face} of $\mathcal{C}$, if for any element $\alpha \in \mathcal{C}$, and for all decompositions, $\alpha = \alpha_1 + \alpha_2$ such that both $\alpha_1, \alpha_2 \in \mathcal{C}$, it implies that $\alpha_1, \alpha_2 \in \mathcal{F}.$ 
\end{definition}

This means that any element in the face of a convex cone has a possible conic decomposition only into the elements of the face itself. An extremal ray of the cone is a face of dimension $1$. 

For the PSD cone, it is well-known that the faces have a one-to-one correspondence with the subspaces of $\mathbb{C}^d.$ This result can be found in \cite[Corollary 12.4]{barvinok2025course}.

\begin{theorem}
\label{thm:psd-faces}
The following statements are true :
\begin{enumerate}
    \item Let $L$ be a linear subspace of $\C{d}$, then, 
\begin{equation}
    \mathcal{F}_L := \{A \in \mathcal{M}^+_d \mid  \operatorname{ran}(A) \subseteq L\}
\end{equation}

is a face of the PSD cone.

\item Any face of the PSD cone is of the form $\mathcal{F}_L$ for some linear space L.

\item Let $X \in \mathcal{M}^+_d.$ Then, $F_{\operatorname{ran}(X)}$ is the smallest face containing $X$.
\end{enumerate}
\end{theorem}
\noindent \underline{Rank-1 generated cones}: Let $\mathcal{L} := \{M_i \in \mathcal{M}^{\operatorname{sa}}_d\}^N_{i=1}$ be a finite set of complex self-adjoint matrices of size $d \times d$. Let us define the convex cone of PSD matrices with linear \emph{half-space} constraints, 
\begin{equation*}
    \mathcal{M}^+_{d}(\mathcal{L}) := \{ X \in \mathcal{M}^+_{d} \mid  \,  \forall  \, M_i \in \mathcal{L}, \, \braket{M_i,X} \geq 0 \}.
\end{equation*}

To deal with the case of $\braket{M,X} = 0,$ we can also impose the constraint $-M$ along with $M$, since, 
\begin{equation*}
    \braket{M,X} \geq 0 \text{ and }  \braket{M,X} \leq 0 \iff \braket{M,X} = 0.
\end{equation*}
If $\mathcal{L}$ is the empty set, $\mathcal{M}^+_{d}(\mathcal{L})$ corresponds to the full set of PSD matrices $\mathcal{M}^+_{d}$. Since $\mathcal{M}^+_{d}(\mathcal{L}) \subseteq \mathcal{M}^+_{d}$, any rank-$1$ element of the the $\mathcal{M}^+_{d}(\mathcal{L})$ is an extreme ray of $\mathcal{M}^+_{d}(\mathcal{L})$, as it is extremal in the larger cone $\mathcal{M}^+_{d}$. For $\mathcal{M}^+_{d}$ matrices, we know that \emph{all} the extremal rays are only rank-$1$ elements. This result is not true generally for the cones $\mathcal{M}^+_{d}(\mathcal{L})$, which may have extremal rays of higher rank. 

We now focus on the rank-1 generation property of convex conical subsets in PSD. Any such cone is said to satisfy the \emph{rank-1 generated property} if all the elements in the cone have a conic decomposition into rank-1 elements. Formally, for every convex matrix cone $\mathcal{C} \subseteq \mathcal{M}^+_d$, let $\operatorname{R_1}[\mathcal{C}]$ denote the cone generated by its rank-1 elements.

\begin{equation}
    \operatorname{R_1}[\mathcal{C}] := \operatorname{cone} \{\ketbra{v}{v} \in \mathcal{C}\}.
\end{equation}

\begin{definition}[Rank-1 generated cones]
A matrix cone $\mathcal{C}$ is said to be rank 1 generated if \begin{equation}
\label{eq:rank-1-generated}
        \operatorname{R_1}[\mathcal{C}] = \mathcal{C}.
    \end{equation}
\end{definition} 

Clearly, the PSD cone satisfies \cref{eq:rank-1-generated}, and is rank-1 generated. 

We consider an example of a conic subset of $2 \times 2$ PSD matrices in which the rank-1 generated property does not hold. Recall that $2 \times 2$ PSD matrices with trace $1$ can be mapped to the bloch sphere in $\mathbb{R}^3$ with the coordinates, \begin{align*}
r_x &:= \braket{X, \sigma_x}, \\
r_y &:= \braket{X, \sigma_x}, \\
r_z &:= \braket{X, \sigma_z}.
\end{align*} Let us consider the following set of self-adjoint matrices, $$\mathcal{L} := \{\sigma_z, -\sigma_z, \sigma_x, \sigma_y \}.$$ 

For the set $\mathcal{M}^+_d(\mathcal{L}),$ this corresponds to the half-space constraints, \begin{align*}
r_x &\geq 0, \\
r_y &\geq0, \\
r_z &=0.
\end{align*} 

In \cref{fig:bloch-constraints}, we represent the set of matrices satisfying these constraints, $\mathcal{M}^+_d(\mathcal{L})$. Note that the rank-1 elements lie on the boundary of the sphere, and hence, $\operatorname{R}_1[\mathcal{M}^+_d(\mathcal{L})] \subsetneq \mathcal{M}^+_d (\mathcal{L}).$

\begin{center}
\begin{figure*}[htbp] 
    \includegraphics[width=0.5\textwidth]{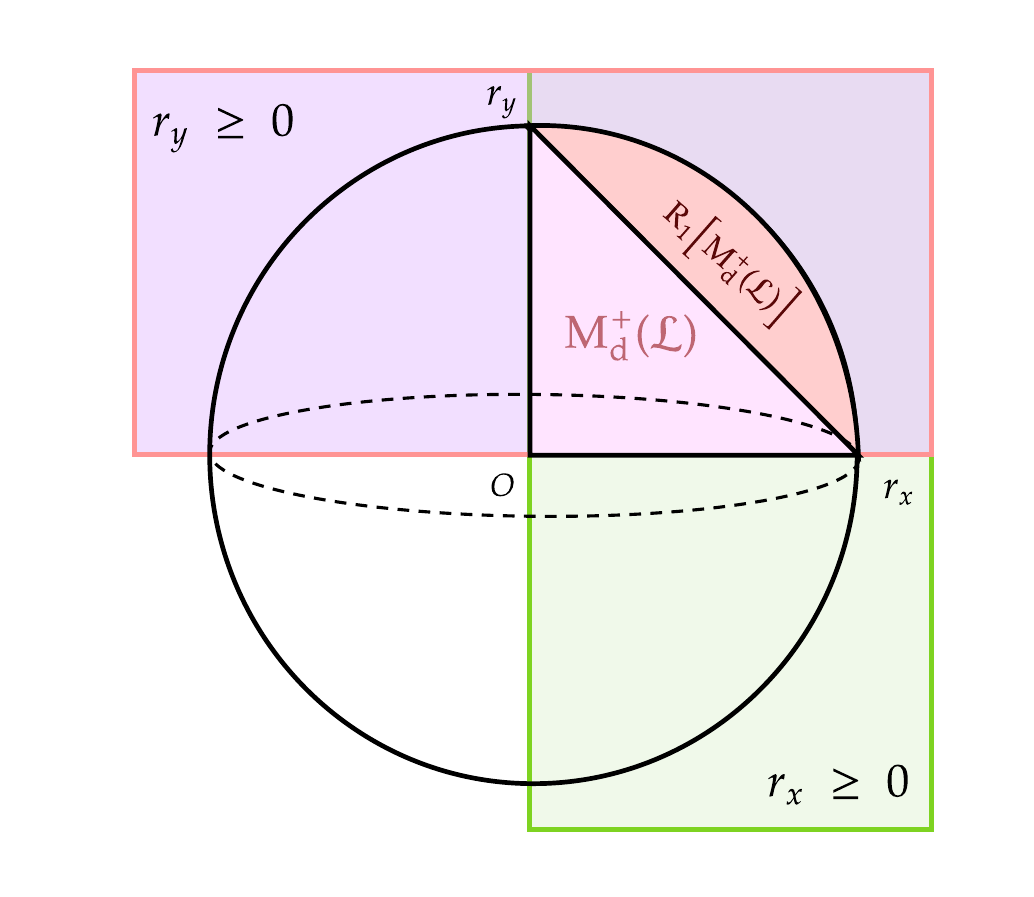}
    \label{fig:bloch-constraints}
    \caption{We plot the slice of the (normalised) Bloch sphere with $r_z = 0$ and represent the constraints $r_x \geq 0$ and $r_y \geq 0$ in (light purple and light green, respectively). The extremal points (points on the boundary of the sphere) are the rank-1 extremal elements. The region in pink is the set $\mathcal{M}^+_d (\mathcal{L})$ and in orange is the $\operatorname{R}_1[\mathcal{M}^+_d (\mathcal{L})]$.}
\end{figure*}
\end{center}

Any closed, convex, and pointed matrix cone is rank-1 generated if and only if all the extremal rays of the cone are rank-1. We phrase it as the next lemma.

\begin{lemma}
A closed convex cone $\mathcal{C} \subseteq \mathcal{M}^+_d$ is rank-1 generated if and only if $\operatorname{ext} \mathcal{C}$ consists of only rank-1 elements. 
\end{lemma}

\begin{proof}
($\implies$) To prove the implication, we assume that $\operatorname{R}_1[\mathcal{C}] = \mathcal{C}$, i.e $\mathcal{C}$ is rank-1 generated. Let there be any element $X$ in $\operatorname{ext} \mathcal{C}$. Since the cone is rank-1 generated, by definition, there exists a decomposition of $X = \sum_k {X_k}$ such that each $X_k$ is rank 1. Also, since $X$ is extremal, we have $X_k \propto X$ for all $k$ which implies $X$ is rank-1. 

($\impliedby$) We assume that the extremal rays of $\mathcal{C}$ are rank-1. Then the proof follows from the fact that a closed convex pointed cone is the convex hull of its extreme rays \cite{rockafellar1970convex}. 
\end{proof}

The rank-1 generated property of matrices is related to problems in optimisation theory, see \cite{argue2023necessary} for more details. 

\subsection{Graph theory}

A (simple and undirected) \emph{graph} $G$ consists of a finite vertex set $V$, and a finite set $E$ of two-element (unordered) subsets of $V$, which is called the edge-set. For $k \geq 3$, a \emph{$k$-cycle} in a graph is a sequence of edges $\{v_0, v_1\}, \{v_1, v_2\}, \dots, \{v_{k-1}, v_k\}$ such that $v_1 \neq v_2 \neq \dots \neq v_k$, and $v_0 = v_k$. 3-cycles in a graph are called \emph{triangles}. A graph which does not contain any triangles is called triangle-free (or triangle-free). A graph is termed \emph{cyclic} if it contains a cycle and \emph{acyclic} otherwise. A graph is called \emph{chordal} if all cycles of size $\geq 4$ have a chord, i.e, an edge connecting two vertices not in the cycle. The \emph{adjacency matrix} $A_G \in \mathcal{M}_d(\{0,1\})$ of a graph $G$ on $d$ vertices is defined as follows: $\operatorname{diag} A_G = 0$ and $(A_G)_{ij} = (A_G)_{ji} = 1$ if $\{i, j\}$ is an edge for $i \neq j$. (Implicitly, $(A_G)_{ij} = (A_G)_{ji} = 0$ otherwise for $i \neq j$). The adjacency matrix is hermitian, and for isomorphic graphs $G$ and $H$, the adjacency matrices are related as $P \, (A_G) \, P^* = A_H$. 
where $P$ is a permutation matrix.

\begin{center}
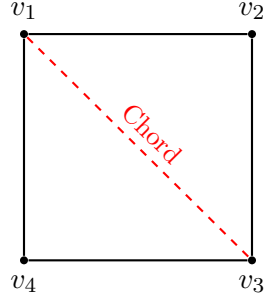
\begin{figure}[h!]
\caption{The figure represents a cycle of the graph $v_1 \rightarrow v_2 \rightarrow v_3 \rightarrow v_4  \rightarrow v_1$ and $v_1 \rightarrow v_3$ is the called the chord of the graph. }
\centering

\begin{tikzpicture}[
    % Define styles for the graph elements
    vertex/.style={circle, fill=black, inner sep=0pt, minimum size=3pt}, % <-- Changed size from 10pt to 3pt
    cycle_edge/.style={draw, thick},
    chord_edge/.style={draw, red, thick, dashed,font=\small}
]

% 1. Place the vertices of the cycle
% The text is now a 'label' outside the node itself
\node[vertex, label=above:$v_1$] (v1) at (0,3) {};
\node[vertex, label=above:$v_2$] (v2) at (3,3) {};
\node[vertex, label=below:$v_3$] (v3) at (3,0) {};
\node[vertex, label=below:$v_4$] (v4) at (0,0) {};

% 2. Draw the edges of the cycle (v1 -> v2 -> v3 -> v4 -> v1)
\draw[cycle_edge] (v1) -- (v2) -- (v3) -- (v4) -- (v1) -- cycle; % Use 'cycle' for a cleaner close

% 3. Draw the chord connecting two non-adjacent vertices (v1 and v3)
\draw[chord_edge] (v1) -- node[midway, above, sloped, red] {Chord} (v3);

\end{tikzpicture}
\end{figure}
\end{center}

\subsection{Separability, PPT and PPT-edge states}

\begin{definition}[Separability and Entanglement]
    A bipartite positive matrix $\rho \in \M{d} \otimes \M{d}$ is said to be separable 
    if $$\rho = \sum^N_{i=1}\ketbra{v_i}{v_i} \otimes \ketbra{w_i}{w_i}$$ for some finite set of vectors $\ket{v_i}, \ket{w_i} \in \C{d}$, and it is said to be entangled otherwise.
\end{definition}

We denote the convex cone of all separable matrices in $\M{d}\otimes \M{d}$ by $\SEPcone_d$. By the Hahn-Banach hyperplane theorem, it is possible to separate this set from every entangled state using a hyperplane. For any \textit{entangled state} $\rho$ we can find a Hermitian operator $W$ such that 

\begin{itemize}
    \item $\text{Tr}(\sigma W) \geq 0$ for all $\sigma$ in  $\SEP{d},$
 \item $\text{Tr}(\rho W) < 0.$ 
\end{itemize} 

This Hermitian operator is called an entanglement witness. Horodecki's criterion \cite{horodecki1996separability} gives us an operational way to detect entanglement, by finding a positive map $\channel$ such that  $(\operatorname{id}_d \otimes \channel) (\rho)$ is not positive semi-definite. One important positive map in this regard is the transposition map $\text{T}$. The quantum states that are not positive under transposition are entangled, and are called NPT (Non-Positive-Transpose) states. The states satisfying the PPT condition are called PPT states, which also include the set of separable states.  This criterion to verify entanglement is called the PPT (Positivity under Partial Transpose) criterion. In the case of $d = 2$, the reverse implication is also true, i.e., any state that is PPT is also separable \cite{Strmer1963, horodecki1996separability}. This is not true for $d \geq 3$. The states that satisfy the PPT condition but are still entangled (for $d \geq 3$) are called $\PPTcone$ entangled states \cite{Horodecki1997PPTent}.

\begin{definition}[Separability Problem (Informal)]
Given a bipartite density matrix $\rho \in \M{d} \otimes \M{d}$, decide whether $\rho \in$  $\SEPcone_d$ or not. 
\end{definition}

It is well known that the membership problem (and the weak membership problem) for $\SEPcone$ is NP-hard \cite{gurvits2003classical, gharibian2008strong}. Unless P = NP, there is no computationally efficient criterion to decide if a state is separable or entangled. We mention here one such important approach for detecting PPT entangled states using the range of the density matrix. 

\begin{definition}[Range criterion \cite{horodecki1997separability}]
 A bipartite mixed state is said to satisfy the range criterion if there exist vectors $\ket{a_i}, \ket{b_i} \in \mathbb{C}^d$ such that $\operatorname{ran}(\rho) = \operatorname{span}(\ket{a_i} \otimes \ket{b_i})$ and $\operatorname{ran}(\rho^\Gamma) = \operatorname{span}(\ket{a_i} \otimes \ket{\overbar{b_i}}).$
\end{definition} Let $\rho$ be a separable state. 
Then, $$\rho = \sum_k \ketbra{v_k}{v_k} \otimes \ketbra{w_k}{w_k},$$ and its partial transpose,
$$\rho^\Gamma = \sum_k \ketbra{v_k}{v_k} \otimes \ketbra{\overbar{w_k}}{\overbar{w_k}}.$$ Clearly $\ket{v_k} \otimes \ket{w_k}$ span the range of $\rho$ and $\ket{v_k} \otimes \ket{\overbar{w_k}}$ span the range of $\rho^\Gamma$, hence all separable states trivially satisfy this criterion. There has been considerable interest in studying the states that maximally violate this criterion. 

\begin{definition}[\cite{lewenstein2000optimization}]
    A PPT entangled state is called an edge-state if there is no product vector $\ket{a} \otimes \ket{b} \in \operatorname{ran}(\rho)$ such that $\ket{a} \otimes  \ket{\overbar{b}} \in \operatorname{ran}(\rho^\Gamma)$.
\end{definition} 

These edge states lie on the boundary of the PPT cone. There is another equivalent characterization of edge states that will be useful for constructing new examples of such states.  

\begin{proposition}[\cite{lewenstein2000optimization}]
    A PPT entangled state $\rho$ is an edge-state if for all (un-normalized) $\ket{ab} :=\ket{a} \otimes \ket{b}$, $\ket{ab} \neq 0$ that $$\rho - \ketbra{ab}{ab} \notin \operatorname{PPT}.$$ 
\end{proposition}

\begin{remark}
    There exist particular faces of quantum states well-suited to the application of the range-criterion, i.e, the set of states that have the range only in \emph{entangled} subspaces. This face can be constructed using ``Unextendible Product Basis", or UPBs. Any PPT state that is orthogonal to the UPBs is automatically PPT entangled as it violates the range criterion \cite{divincenzo2003unextendible, terhal2001family,Bravyi2004Unextendible}.
\end{remark}

It is clear from the definition that all extremal PPT states are edge-states. But it is known that not all edge states are extremal \cite{Lewenstein2000edge}. There are only a few systematic constructions of PPT entangled edge states \cite{choi2020entangled} known in the literature. The ranks of PPT entangled edge states have been amply scrutinized, and a complete classification obtained in $3 \otimes 3$ \cite{kye2012classification}. Such states are also useful to construct entanglement witnesses \cite{lewenstein2000optimization}, and any PPT state can be decomposed into a PPT edge state and a separable state \cite{lewenstein2000optimization}, making them useful to analyse from the perspective of entanglement theory.

\subsection{Local Diagonal Orthogonal Invariant States}

In this section, we provide a small introduction to local diagonal orthogonal invariant states, or in short, LDOI states. These states are characterized by their local invariance, 
$$(O \otimes O) X (O \otimes O) = X$$ for all $O$ in the diagonal orthogonal matrices. This is the group of matrices, 
$$\mathcal{DO}_d := \{\operatorname{diag}(x_1,x_2 \ldots x_n) \mid x_{i} \in \{-1,1\}\}.$$ These are exactly the states parametrized by a matrix triple $(X,Y,Z)$ with equal diagonals. 
\begin{align*}
\label{eq:ldoi-state}
    \rho_{X,Y,Z}  := \sum_{i,j}X_{ij}\ketbra{ij}{ij} &+ \sum_{i \neq j}Y_{ij}\ketbra{ii}{jj} \\ &+\sum_{i \neq j}Z_{ij}\ketbra{ij}{ji}.
\end{align*}

In particular, their PSD and PPT properties can be readily inferred from the matrices $(X, Y, Z).$ The matrix is PSD if and only if $Y$ is self-adjoint, $\forall i,j \in [d], \, \, X_{ij} \geq |Z_{ij}|^2$ and $Y \in \mathsf{PSD}_d.$ For the PPT property, we have the following result,
\begin{theorem}[\cite{singh2021diagonal}]
\label{thm:ppt-ldoi}
    The family of quantum states $\rho_{X,Y,Z}$ is PPT if and only if,  
    \begin{itemize}
        \item $\forall i,j \in[d] \quad X_{ij} \geq |Y_{ij}|^2, |Z_{ij}|^2$,
        \item $Y, Z \in \mathsf{PSD}_d.$
    \end{itemize}
\end{theorem}
\begin{theorem}[\cite{singh2021diagonal}]
\label{thm:sep-ldoi}
    The family of quantum states $\rho_{X, Y, Z}$ is separable if and only if there exists a finite set of vectors $\{\ket{v_k},\ket{w_k}\}_{k\in I}\subset \C{d}$ such that
\begin{align*}
      &X = \sum_{k\in I} |v_k \odot \overbar{v_k}\rangle\langle w_k \odot \overbar{w_k}| \\ &Y = \sum_{k\in I} \ketbra{v_k \odot w_k}{v_k \odot w_k}, \\ &Z= \sum_{k\in I} \ketbra{v_k \odot \overbar{w_k}}{v_k \odot \overbar{w_k}}.
\end{align*}
\end{theorem}

In particular, deciding if a LDOI state is separable or entangled is still NP-hard \cite{singh2021diagonal}, and there exist PPT entangled states. For properties of LDOI states and the separability problem, we recommend \cite{singh2021diagonal, singh2020entanglement, gulati2025entanglement}. Even in the smallest non-trivial $3 \otimes 3$ system, there is still no complete characterization of entanglement available for these states. Finally, we show that the popular CCNR (computable cross norm or realignment) criterion for LDOI states can be understood as a positivity condition on the triple of matrices.

\begin{theorem}[\cite{singh2021diagonal}]
\label{thm:ccnr-ldoi}
    Let $||X||_1 = \sum^d_{i,j=1} |X_{ij}|$ and $||X||_{\operatorname{tr}} = \operatorname{tr}(\sqrt{X^*X}).$ Then the LDOI state $\rho_{X,Y,Z}$ satisfies the CCNR  criterion if and only if 
    $$||X||_1 - ||X||_{\operatorname{tr}} \geq 2 \sum_{1\leq i < j \leq d} \operatorname{max} |B_{ij}|, |C_{ij}|.$$
\end{theorem}

\section{Main results}
\label{sec:main-results}
Let $\rho$ be a mixed (but un-normalised) bipartite quantum state in $\M{d} \otimes \M{d}$. Our goal is to provide the necessary conditions for separability using the rank-1 generated property. Note that a rank-1 PSD matrix $\ketbra{v}{v}$ is PPT if and only if $\ket{v}= \ket{v_1} \otimes \ket{v_2}$. Hence, the cone of separable states, or SEP, is exactly equal to $\operatorname{R}_1[\text{PPT}]$, that is, the rank-1 generated matrices in PPT states. 

The constraint of PPT is \emph{not} a linear matrix constraint. We want to look at mappings of these cones to $n\times n$ matrices that allow us to reduce the PPT condition to \emph{linear} constraints. Let us state our main idea as the next theorem.

\begin{theorem}
\label{thm:rank-1-generated-ent}
For all $K \in \mathcal{M}_{n,d^2}$, we define the following positive map $Z_K : \mathcal{M}_{d} \otimes \mathcal{M}_d \cong \mathcal{M}_{d^2} \rightarrow \mathcal{M}_{n}$
    \begin{align*}
        Z_K(\rho) := K\rho K^*.
    \end{align*} Let $\mathcal{F}$ be a face of the bipartite psd matrices $(\mathcal{M}_d \otimes \mathcal{M}_d)^+$ and $\mathcal{F}_{PPT} := \mathcal{F} \, \cap \, \mathsf{PPT}$. Furthermore, assume that $Z_K(\mathcal{F}_{PPT}) \in \mathcal{M}^+_n(\mathcal{L})$. Then the following statements are true. 

\begin{enumerate}
    \item We have for all separable states on $\mathcal{F}$, the matrix $Z_K(\rho) \in \operatorname{R}_1[\mathcal{M}^+_n(\mathcal{L})]$. This is a \emph{necessary criterion} for separability for states in $\mathcal{F}$.
    
    \item If for a PPT state on $\mathcal{F}$ such that, $Z_K(\rho) \in \operatorname{ext} \mathcal{M}^+_n(\mathcal{L})$ and has rank $\geq 2$, the state $\rho$ is a PPT entangled edge state.

\end{enumerate}

\end{theorem}

\begin{proof} 

Firstly, note that $Z_K$ maps positive matrices into positive matrices, and rank-1 matrices to rank-1 matrices as $Z_K(\ketbra{v}{v}) = (K\ket{v})(K\ket{v})^*$. If $\rho$ is a separable state, it has a decomposition into rank-1 states that are PPT. 
\begin{equation}
    \rho = \sum^N_{k=1} \ketbra{v_k}{v_k} \otimes \ketbra{w_k}{w_k}. \end{equation} 
Define $\rho_k :=\ketbra{v_k}{v_k} \otimes \ketbra{w_k}{w_k},$ that is a rank-1 PPT state. Moreover, if we assume that $\rho$ is separable and lies on the face $\mathcal{F}$, it implies that each $\rho_k \in \mathcal{F}$, hence by assumption in the theorem, each $Z_K(\rho_k) \in \mathcal{M}^+_d(\mathcal{L})$. Therefore, combined with the rank-1 preserving quality of the map, we conclude that $Z_K(\rho) = \sum^N_{k=1} Z_K(\rho_k) \in \operatorname{R}_1[\mathcal{M}^+_d(\mathcal{L})]$ as $\rho_k$ and hence $Z_K(\rho_k)$ is rank-1. This concludes the proof of the first part of the theorem. For the second part, note that for any product state $\ket{ab}$, the matrix $Z_K (\ketbra{ab}{ab}) = \ketbra{v}{v} \in \mathcal{M}^+_d(\mathcal{L})$. Let $\rho$ satisfy the condition in part 2. Assume the converse, that there is some product state $\ket{ab}$ such that $\rho - \ketbra{ab}{ab} = \delta \in \operatorname{PPT}$. Then $\rho = \ketbra{ab}{ab} + \delta \implies Z_K(\rho) = \ketbra{v}{v} + Z_K(\delta)$. Since $Z_K(\rho)$ is extremal, $Z_K(\rho) \propto \ketbra{v}{v}$, that violates the fact that it is rank $\geq 2$.
\end{proof}

We also state the analogous theorem that maps the \emph{partial transpose} of the state. The proof of this theorem is similar to that of the \cref{thm:rank-1-generated-ent}.
\begin{theorem}
\label{thm:rank-1-generated-ent-transpose}
For all $K \in \mathcal{M}_{n,d^2}$, we define the following positive map $Z^\Gamma_K : (\mathcal{M}_{d} \otimes \mathcal{M}_d) \cong \mathcal{M}_{d^2} \rightarrow \mathcal{M}_{n}$
    \begin{align*}
        Z^\Gamma_K(\rho) := 
        (Z_K \circ \Gamma)(\rho) =  K\rho^\Gamma K^*
    \end{align*} Let $\mathcal{F}$ be a face of the bipartite psd matrices in $\mathcal{M}_d \otimes \mathcal{M}_d$ and $\mathcal{F}_{PPT} := \mathcal{F} \, \cap \, \mathsf{PPT}$. Furthermore, assume that $Z^\Gamma_K(\mathcal{F}_{PPT}) \in \mathcal{M}^+_n(\mathcal{L})$. Then the following statements are true. 

\begin{enumerate}
    \item We have for all separable states on $\mathcal{F}$, the matrix $Z^\Gamma_K(\rho) \in \operatorname{R}_1[\mathcal{M}^+_n(\mathcal{L})]$. This is a \emph{necessary criterion} for separability for states in $\mathcal{F}$.
    
    \item If for a PPT state on $\mathcal{F}$, $Z^\Gamma_K(\rho) \in \operatorname{ext} \mathcal{M}^+_n(\mathcal{L})$ and has rank $\geq 2$, the state $\rho$ is a PPT entangled edge state.

\end{enumerate}

\end{theorem}
\subsection{Detection of PPT entangled quantum states}
\label{sec:entanglement-detection}
We use the following strategy based on the \cref{thm:rank-1-generated-ent} to detect PPT entangled states.

\begin{enumerate}
    \item Let $\mathcal{F}$ be a face of the bipartite $(\mathcal{M}_{d} \otimes \mathcal{M}_d)^+$ matrices. 
    \item We choose $K$ such that $Z_K(\mathcal{F}_{PPT}) \in \mathcal{M}^+_n(\mathcal{L})$ or $Z^\Gamma_K(\mathcal{F}_{PPT}) \in \mathcal{M}^+_n(\mathcal{L})$. 
    \item By \cref{thm:rank-1-generated-ent}, when the state $\rho$ is separable, $Z_K(\rho) \in \operatorname{R}_1[\mathcal{M}^+_n(\mathcal{L})]$ or $Z^\Gamma_K(\rho) \in \operatorname{R}_1[\mathcal{M}^+_n(\mathcal{L})]$ (resp.).
    \item We use the gap in the inclusion $\operatorname{R}_1[\mathcal{M}^+_n(\mathcal{L})] \subsetneq \mathcal{M}^+_n(\mathcal{L})$ to detect PPT entangled states. 
\end{enumerate}

In some cases that we consider in this article, this problem of rank-1 generated cones has been well-studied in the literature, even though the connection to PPT entanglement has not been considered.

\subsubsection{Sparse States}
\label{sec:psd-graphs}
For a $n \times n$ matrix $X $, we define the (undirected) graph of a matrix $G(X)$ with vertices in $[n]$ and the edge set $\{\{i,j\} \, | \,  X_{ij} \cdot X_{ji} > 0\}$, i.e the graph has an edge whenever $X_{ij}$ and $X_{ji}$ are both non-zero. This allows us to define the linear space of self-adjoint matrices with a sparsity pattern given by an undirected graph $H$, 
\begin{equation*}
    \mathcal{M}_{n}(H) := \{X \in \M{n} \, |  \, X^* = X, \, G(X) \subseteq H\}
\end{equation*} Note that for the complete graph on $d$ vertices, which we denote by $K_n$, this space $\mathcal{M}_{n}(K_n)$ is the real vector space of $n \times n$ self adjoint matrices. In the other extreme case, when the graph is empty (denoted by $E_n$), $\mathcal{M}_{n}(E_n)$ is the vector space of diagonal matrices with real entries. We denote the set of sparse PSD matrices \cite{agler1988positive, paulsen1989Schur} as $$\mathcal{M}_n^+(H) := \mathcal{M}^+_{n} \, \cap \, \mathcal{M}_{n}(H)$$ Let $\{\ket{i_A}\}^n_{i=1}$ and $\{\ket{j_A}\}^n_{j=1}$ be two orthogonal sets of vectors of size $n$ each. We define the following $n \times n$ matrix $D(\rho)$ entrywise as 
    $$D(\rho)_{ij} := \bra{i_A j_B}\rho\ket{i_A j_B}.$$

To understand the sparsity pattern  of this matrix $D(\rho)$, we look at the graph $G(D(\rho))$. Note that elements of $D(\rho)$ are just the diagonal elements of the density matrix. Then, we define the face of the cone using the sparsity pattern of matrix $A$, and a graph $H$ on $n$ vertices. 

\begin{equation}
    \mathcal{F}[H]:= \{\rho \in (\mathcal{M}_d \otimes \mathcal{M}_d)^+ \, | \, G\big(D(\rho)\big) \subseteq H\}
\end{equation}

For the case of triangle-free graphs in \cite{singh2020entanglement}, this class of states was referred to as triangle-free quantum states. Let us quickly show that this is a face of the cone. Firstly, note that $D(\rho)$ is a non-negative matrix for all PSD matrices $\rho$. 
Let us pick an element $\rho \in \mathcal{F}_H$, and write a decomposition into PSD matrices, 
$$\rho = \rho_1 + \rho_2 \implies D(\rho) = D(\rho_1) + D(\rho_2)$$ If $D(\rho)_{ij} = 0$, it implies that $k = \{1,2\},\, D(\rho_i)_{ij} = 0$, and hence $G\big(D(\rho_{k})\big) \subseteq G\big(D(\rho)\big) \subseteq H$, showing that $\rho_1, \rho_2 \in \mathcal{F}_H$. We show that for all PPT states on this face, the following constraints hold.  

\begin{theorem}
\label{thm:sparse-PSD}
    Let $\rho \in \mathcal{F}_H$ be a bipartite quantum state satisfying the PPT condition, then, by choosing,  
    \begin{itemize}
        \item $K := \sum^n_{i=1}\ketbra{i}{i_Ai_B} \implies Z_{K}(\rho) \in \mathcal{M}_n^+(H)$
        \item $K := \sum^n_{i=1}\ketbra{i}{i_A\overbar{i_B}} \implies Z^\Gamma_{K}(\rho) \in \mathcal{M}_n^+(H)$ 
    \end{itemize}
\end{theorem}

\begin{proof}
     Let $\rho \in \mathcal{F}_H$ be a bipartite quantum state satisfying the PPT condition. We have, 
     \begin{align*}
     Z_{K}(\rho) = \sum^n_{i,j=1}\braket{i_A i_B |\rho |j_A j_B} \ketbra{i}{j}.
     \end{align*}
    For any PPT state, $\rho^\Gamma$ is PSD. Hence, any $2 \times 2$ block of $\rho^\Gamma$ is PSD. This implies for all $i,j$, 
    \begin{align*}
    \braket{i_A \overbar{j_B}|\rho^\Gamma |i_A \overbar{j_B}}\braket{j_A \overbar{i_B}|\rho^\Gamma |j_A \overbar{i_B}} &\geq \abs{\braket{i_A \overbar{j_B}|\rho^\Gamma |j_A \overbar{i_B}}}^2 \\
    \braket{i_Aj_B|\rho|i_Aj_B}\braket{j_Ai_B|\rho|j_Ai_B} &\geq \abs{\braket{i_A i_B|\rho |j_A j_B}}^2
    \\
    \iff D(\rho)_{ij} D(\rho)_{ji} &\geq |Z_K(\rho)|^2_{ij} .
    \end{align*}
Therefore, we have $G(Z_K(\rho)) \subseteq G(D(\rho)) \subseteq G(H).$ The proof of the second statement is similar.
\end{proof}

Note that the intersection $\mathcal{M}^+_{n}(H)$ is a convex cone with constraints encoded in the sparsity pattern of this matrix using the graph $H$. Now, by using our main theorem, \cref{thm:rank-1-generated-ent} and applying it to the previous result, \cref{thm:sparse-PSD}, we can derive the following necessary criterion for separability.

\begin{corollary}
\label{cor:sparse-psd-rank-1}
Let $K := \sum^n_{i=1}\ketbra{i}{i_Ai_B}.$ Let $\rho \in \mathcal{F}_H$ be a separable quantum state, then, 
     \begin{equation}
        Z_{K}(\rho)\in \operatorname{R}_1 [\mathcal{M}^+_{n}(H)]
     \end{equation}
\end{corollary}

Therefore, we find that the matrix $\sum^n_{i,j=1}\braket{i_A i_B |\rho |j_A j_B} \ketbra{i}{j}$ is rank-1 generated in $\mathcal{M}^+_{n}(H)$ for all separable states in the face of sparse states, $\mathcal{F}_H.$ For the empty graph, $E_n$ on $n$ vertices, $\operatorname{R}_1[\mathcal{M}^+_{n}(E_n)]$ is the set of diagonal matrices with positive entries. On the other extreme, for the complete graph, which we denote as $K_n$, we  have $\operatorname{R}_1[\mathcal{M}^+_{n}(K_n)] = \mathcal{M}^+_{n} = \mathcal{M}^+_{n}(K_n)$ is the set of PSD matrices. In both of these cases, the cone is equal to its rank-1 generated sub-cone. There is actually a complete characterization of all such graphs. 

\begin{theorem}[\cite{paulsen1989Schur}]
The cone $\mathcal{M}^+_{n}(G)$ is rank-1 generated, i.e  $\operatorname{R}_1[\mathcal{M}^+_{n}(G)] = \mathcal{M}^+_{n}(G)$ if and only if the graph is chordal. 
\end{theorem}

For all non-chordal graphs, these results allow us to detect entangled states by checking that the matrix $Z_K{(\rho)}$ is not rank-1 generated. Although in general a complex problem, a semi-definite characterisation of this rank-1 generated property can be shown for triangle-free graphs based on the so-called comparison matrices. This characterization has also recently been made in \cite{singh2020entanglement}. For every matrix $X$, we can define the comparison matrix $M(X)$ entrywise such that $[M(X)]_{ii} = |X_{ii}|$ and  $[M(X)]_{ij} = -|X_{ij}|$ for all $i \neq j$. Then, we have the following result.  
\begin{theorem}[Triangle-free graphs]
\label{thm:triangle-free-graphs}
    Assume $G$ is a triangle free graph, and $X \in \mathcal{M}^+_{n}(G)$. Then it is true that, 
    $$X \in \operatorname{R}_1[\mathcal{M}^+_{n}(G)] \iff M(X) \geq 0$$
\end{theorem}

\begin{proof}
    The proof of the implication follows the strategy in \cite[Theorem III.2]{singh2020entanglement}. For any rank 1 element $\ketbra{v}{v} \in \mathcal{M}^+_n(G)$, the $\operatorname{supp}(v)$ has to be a clique of the graph $G$. Therefore, for triangle-free graphs, since we only have the largest cliques of size $2$, we have $|\operatorname{supp}(v) |\leq 2$. Therefore, for all $X \in \operatorname{R}_1[\mathcal{M}^+_{n}(G)]$, i.e rank-1 generated elements, we have (denote $\operatorname{supp}(v) =: \sigma(v)$)
    \begin{equation}
        X = \sum_{\sigma (v_k) = i_k} \ketbra{v_k}{v_k} + \sum_{\sigma (z_k) = \{i_k, j_k\}} \ketbra{z_k}{z_k}
    \end{equation}
    and hence the comparison matrix can be computed as, 
    \begin{equation}
    M(X) = 
    \sum_{\substack{\sigma(v_k) = i_k}} 
    \ketbra{v_k}{v_k}
    + 
    \sum_{\substack{\sigma(z_k) = \{i_k, j_k\}}} 
    M\left( \ketbra{z_k}{z_k} \right)
\end{equation}
For any matrix with a support on a $2 \times 2$ block, $B \geq 0 \iff M(B) \geq 0$. Hence, from this we conclude that $M(\ketbra{z_k}{z_k}) \geq 0$, and hence also $M(X) \geq 0$. This shows the implication. To show the other direction, we use the result from \cite[Theorem 3.9]{singh2020ppt2} which states that for any matrix $X \in \mathcal{M}^+_n$, $M(X) \geq 0 \iff X = D B D$ where $B \in \mathcal{M}^+_n$ is a diagonally dominant matrix such that $G(B) = G(X)$ and $D$ is a strictly positive diagonal matrix. 
    Then by expanding, $B = \sum_{\{i, j\} \in E(G)} B^{(i,j)} + C$ where \begin{align*}
    B^{(i,j)} = \begin{bmatrix} |B_{ij}| & B_{ij} \\ B_{ij}^* & |B_{ij}| \end{bmatrix}_{i,j} &= |B_{ij}| \big(\ket{i} \bra{i} + \ket{j}\bra{j}\big) \\ &+ B_{ij} \ket{i}\bra{j} + B^*_{ij} \ket{j} \bra{i}
    \end{align*}and $C$ is a diagonal matrix with $C_{ii}:= B_{ii} - \sum_{i \neq j} |B_{ij}|.$ Note that $C_{ii} \geq 0$ as $B$ is diagonally dominant. Also for all $B^{(i,j)}$, we have $\operatorname{det} \Big(B^{(i,j)}\Big) = 0$, which implies it is rank-1 and hence $B^{(i,j)} = \ketbra{z_{ij}}{z_{ij}}$ for $z_{ij} \in \mathbb{C}^n$ with the support $\{i,j\} \in E.$  Now for any positive-diagonal matrix $D$,
    $D\ketbra{z_{ij}}{z_{ij}}D \in \mathcal{M}^+_n(G)$ as $D \ket{z_{ij}}$ has support $\{i,j\} \in E$. Moreover, $DCD$ is a diagonal matrix, and hence trivially belongs to $\mathcal{M}^+_n(G).$ 
\end{proof}

\begin{center}

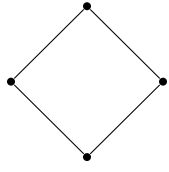
\begin{figure}
\centering
\begin{tikzpicture}[
    % Define a style for the vertices
    vertex/.style={
        circle,         % Shape
        fill=black,     % Fill color
        inner sep=0pt,  % No padding inside the shape
        minimum size=3pt % Diameter of the dot (adjust as needed)
    }
]

\node[vertex] (v1) at (90:1cm)  {}; % Top vertex
\node[vertex] (v2) at (180:1cm) {}; % Left vertex
\node[vertex] (v3) at (270:1cm) {}; % Bottom vertex
\node[vertex] (v4) at (0:1cm)   {}; % Right vertex

% Draw the edges connecting the vertices in a cycle
% The '-- cycle' command automatically connects the last node (v4) back to the first (v1)
\draw (v1) -- (v2) -- (v3) -- (v4) -- (v1);

\end{tikzpicture}

\caption{The graph $\mathcal{C}_4$ is the smallest non-chordal graph. It is also triangle-free, hence for any state with $G(A) \subseteq \mathcal{C}_4$, we can apply the theorem \cref{thm:triangle-free-graphs} to detect entanglement}

\label{fig:graph-non-chordal}
\end{figure}
\end{center}

\begin{corollary}
\label{cor:induced-graph}
Assume that graph $H$ on vertex set $I_H$ is a triangle-free induced subgraph of $G$. Then, for every $X \in \operatorname{R}_1[\mathcal{M}^+_{|I_H|}(G)]$, we have 
$$M(X[I_H]) \geq 0$$
\end{corollary}

\begin{proof}
    It is easy to see that $X[I_H] \in \operatorname{R}_1[\mathcal{M}^+_{|I_H|}(H)]$, and hence the result follows from \cref{thm:triangle-free-graphs}.
\end{proof}

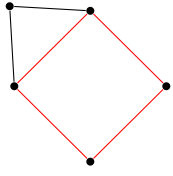
\begin{figure}
    \centering
    \begin{tikzpicture}[
    % Style for vertices (small black dots)
    vertex/.style={
        circle,
        fill=black,
        inner sep=0pt,
        minimum size=3pt
    },
    % Style for the cycle edges
    cycle_edge/.style={
        draw=red % Set color to red
    }
]

% Define the 5 vertices
% We'll place the 4-cycle vertices first, similar to before
\node[vertex] (v1) at (90:1cm)  {}; % Top vertex of the cycle
\node[vertex] (v2) at (180:1cm) {}; % Left vertex of the cycle
\node[vertex] (v3) at (270:1cm) {}; % Bottom vertex of the cycle
\node[vertex] (v4) at (0:1cm)   {}; % Right vertex of the cycle

% Define the 5th vertex to form the triangle.
% Let the common edge be between v1 and v2.
% We place v5 somewhat 'outside' this edge.
\node[vertex] (v5) at (135:1.5cm) {}; % Positioned relative to v1 and v2

% Draw the edges
% First, draw the 4-cycle edges in red.
% The edge v1--v2 is part of the cycle and also the common edge.
\draw[cycle_edge] (v1) -- (v2) -- (v3) -- (v4) -- (v1);

% Now, draw the remaining edges of the triangle (v1--v5 and v2--v5).
% These edges are not part of the original 4-cycle, so they use the default (black) color.
% The edge v1--v2 is already drawn in red.
\draw (v1) -- (v5) -- (v2); % This draws v1--v5 and v5--v2

\end{tikzpicture}
    \caption{The graph in the figure is not triangle-free, but it has a triangle-free induced subgraph (the subgraph in red), allowing the application of \cref{cor:induced-graph}}.
\end{figure}

Since every non-chordal graph contains a graph which is triangle-free (because it contains a cycle without chords), these tests provide an easy-to-implement family of tests for inclusion into the rank-1 generated cone (and also, novel necessary conditions on separability in the face of quantum states). We want to point out that these tests are not complete, in the sense that some matrices that are not included in $\operatorname{R}_1[\mathcal{M}^+_{n}(G)]$ possibly satisfy all the conditions in \cref{cor:induced-graph}. The complete family of matrix maps that characterize the inclusion into the rank-1 generated cone of sparse PSD matrices has been completely characterized in \cite{paulsen1989Schur}. These tests are based on checking positivity of the matrix after application of certain Schur maps, and this apparatus is intimately connected to the problem of positive-semidefinite completions \cite{grone1984positive}. 
For any matrix $Z \in \mathcal{M}_d$, we define the map $\mathcal{T}_Z (X) = Z \odot X$ where $\odot$ is the entrywise (or Schur) product of matrices $(Z \odot X)_{ij} = Z_{ij} X_{ij}$. It is well known that this Schur product map is a  positive map if and only if $Z \in \mathcal{M}^+_n$. Recall that a matrix is PSD if and only if all its principal submatrices are PSD matrices. We extend this to define the convex cone of matrices that have \emph{few} principal minors PSD. To do this, recall that a clique of a graph is a subgraph that is complete. For any graph $G$, let $\operatorname{cliq}[G]$ denote the set of cliques. For any matrix in $X \in \mathcal{M}_n$, and $I \subseteq [n]$, we define the principal submatrix $X[I] \in \mathcal{M}_{|I|}$ on the indices in $I$. Then, we define the following notion, 

\begin{definition}
    For a graph $G$,  
\begin{equation*}
    \mathcal{M}^*_n(G) := \{Z \in \mathcal{M}_n \, \mid \,  \forall I \in \operatorname{cliq}[G], \, \, Z[I] \geq 0 \}.
\end{equation*}
\end{definition}

Note that for a complete graph, i.e, $G = K_n$, this reduces to the PSD matrix cone, as \emph{all} the principal submatrices are PSD. We show that this is the dual notion to the rank-1 generated cone $\operatorname{R}_1[\mathcal{M}^+_n(G)],$ in the next theorem.

\begin{theorem}
\label{thm:dual-sparse-psd}
Let $X \in \mathcal{M}^+_n(G).$ Then the following is true, 
    $$X \in \operatorname{R}_1[\mathcal{M}^+_n(G)] \iff \mathcal{T}_Z (X)\geq 0$$ for all $Z \in \mathcal{M}^*_n(G)$. 
\end{theorem}

\begin{proof}
    For the forward implication, let $\ketbra{x}{x} \in \mathcal{M}^+_n(G).$ Then, $\ket{x}$ must have the support only on a clique of the graph $G,$ say $I$, and therefore, $\ketbra{x}{x}$ must be supported in $I.$ Moreover, we claim that $\mathcal{T}_Z (\ketbra{x}{x}) = \ketbra{x}{x} \odot Z \geq 0$ for all $Z \in \mathcal{M}^*_n(G)$. This holds because $Z[I] \geq 0$. This proves the implication. For the reverse implication, we refer the reader to \cite[Lemma 2.3]{paulsen1989Schur}.
\end{proof}

Therefore, it directly follows that for $X \notin \operatorname{R}_1[\mathcal{M}^+_n(G)]$, there exists a $Z \in \mathcal{M}^*_n(G)$ such that  $\mathcal{T}_Z (X) \ngeq 0$. By \cref{cor:sparse-psd-rank-1}, this also detects entanglement in the class of sparse states we consider. This result also allows us to provide another proof of the forward implication in \cref{thm:triangle-free-graphs}.
\begin{corollary}
    For a triangle-free graph, $G$, $X \in\mathcal{M}^+_n(G) \implies M(X) \geq 0$. 
\end{corollary}

\begin{proof}
    Choose $Z_{ij} = e^{-i \operatorname{arg}(X_{ij})}$. We show that this matrix belongs to $\mathcal{M}^*_n(G)$. Since the graph is triangle-free, all its cliques are of size $2$. Hence, the matrix $Z[i,j] \geq 0$ for every $2 \times 2$ principal minor of the matrix. Applying \cref{thm:dual-sparse-psd}, we get $\mathcal{T}_Z(X) = M(X) \geq 0$.
\end{proof}

Let us proceed with defining the notion of PSD-completable matrices. A matrix $X \in \mathcal{M}^*_n(G)$ is called to be PSD-\emph{completable} if there exists a matrix $\Tilde{X} \in \mathcal{M}^+_n$ such that $\Tilde{X}[I] = X[I]$ for all $I \in \operatorname{cliq}[G].$

\begin{theorem}
    $\mathcal{T}_Z (X)\geq 0$ for all $X \in \mathcal{M}^+_n(G)$ if $Z$ is PSD-completable. Therefore $X \notin \operatorname{R}_1[\mathcal{M}^+_n(G)]$ can be detected only by non-completable matrices. 
\end{theorem}

\begin{proof}
    Let $Z'$ be the completion of the matrix $Z$. Then $\mathcal{T}_Z(X) = Z \odot X = Z' \odot X \geq 0$. 
\end{proof}
Note that these witnesses are analogous to decomposable ones, which do not detect PPT entangled states. It is important to note that in this dual picture, we have the analogous theorem for chordal graphs. (See also \cite[Theorem 2.4]{paulsen1989Schur}).

\begin{theorem}
The following are equivalent for a graph $G$, 
\begin{enumerate}
    \item All $Z \in \mathcal{M}_n^*(G)$ are completable. 
    \item The graph is chordal. 
\end{enumerate}
\end{theorem}
Let us summarise our results in the following slogan. \emph{If the graph of the matrix $A$ is non-chordal, then the gap between the sparse PSD cone and its rank-1 generated version detects PPT entanglement}. This result provides a huge diversity for the construction of such states, each of which has $\approx O(d^2)$ free parameters. We show that this is very useful for constructing families of PPT entangled LDOI states.  
Select a \emph{non-chordal} graph $G$ on $d$ vertices. Take $Y \notin \operatorname{R}_1[\mathcal{M}^+_n(G)]$, and select $Z \geq 0$ such that diag(Z) = diag(Y) and $G(Z) = G$. Then, select a non-negative $X$ (ensuring diag(X) = diag(Y)) such that $G(X) = G$ and  $X_{ij} X_{ji} \geq |Z_{ij}|^2$.

Then, the following state is PPT entangled.
\begin{align*}
    \rho_{X,Y,Z}  = \sum_{i,j}X_{ij}\ketbra{ij}{ij} &+ \sum_{i \neq j}Y_{ij}\ketbra{ii}{jj} \\ &+\sum_{i \neq j}Z_{ij}\ketbra{ij}{ji}.
\end{align*}

    \begin{proof}[Proof of entanglement]
    The PPT property can be readily verified from \cref{thm:ppt-ldoi}. To check that the state is entangled, note that these states lie on the face $\mathcal{F}_G$, and for $K = \sum^d_{i=1} \ket{i}\bra{ii}$, $Z_K(\rho) \notin \operatorname{R_1}[\mathcal{M}_d(G)].$
    \end{proof}

    \begin{example}
        We select the smallest non-chordal graph, the cycle graph $\mathcal{C}_4$, and we use the extremal matrix of $C$ that has rank $2$ \cite[Example 1]{agler1988positive}:
\[
Y = \begin{bmatrix}
    1 & 1 & 0 & 1 \\
    1 & 2 & 1 & 0 \\
    0 & 1 & 1 & -1 \\
    1 & 0 & -1 & 2 \\
\end{bmatrix}.
\]

Then, by using the following free parameters, \[
X = \begin{bmatrix}
    1 & X_{12} & 0 & X_{14} \\
    X_{21} & 2 & X_{23} & 0 \\
    0 & X_{23} & 1 & X_{34} \\
    X_{41} & 0 & X_{43} & 2 \\
\end{bmatrix}, \] \[ Z = \begin{bmatrix}
    1 & Z_{12} & 0 & Z_{14} \\
    Z_{12}^* & 2 & Z_{23} & 0 \\
    0 & Z_{23}^* & 1 & Z_{34} \\
    Z_{14}^* & 0 & Z_{34}^* & 2 \\
\end{bmatrix}
\]
such that $Z \geq 0$ and $X_{ij}X_{ji} \geq |Z_{ij}|^2$, results in a PPT entangled state. Note that due to extremality, the constructed matrix forms a family of PPT entangled edge-states. 

\subsubsection{Restricted Rank-1 States}
\label{sec:corr-matrices}

In this section, we provide another face of PSD bipartite matrices for which our result in \cref{thm:rank-1-generated-ent-transpose} detects PPT entangled states.  Let $V = \operatorname{span}^n_{i=1}\{\ket{ii}\}$. Then, we can define the following face of quantum states, 

$$\mathcal{F}_{\ket{\phi}} := \{\rho \in (\mathcal{M}_d \otimes \mathcal{M}_d)^+ \mid \restr{\rho}{V} \propto \ketbra{\phi}{\phi}\}.$$ 

This is the face in which all the density matrices restricted to $V$ are proportional to the rank-1 matrix, $\ketbra{\phi}{\phi}.$ It is easy to see why these matrices form a face of the convex cone. Particularly, let's assume $\rho \in \mathcal{F}_{\ket{\phi}}$ has a conic decomposition into sum of two PSD matrices, i.e $\rho =\rho_1 + \rho_2 \mid \rho_1, \rho_2 \in (\mathcal{M}_{d} \otimes \mathcal{M}_d)^+.$ Then, it implies that, 

$$\restr{\rho}{V} = \restr{\rho_1}{V} + \restr{\rho_2}{V}.$$ Since $\restr{\rho}{V} \propto \ketbra{\phi}{\phi}$ that is extremal, we have, $\restr{\rho_1}{V}$ and $ \restr{\rho_2}{V}$ also are $\propto \ketbra{\phi}{\phi}.$ This shows that this is a face of the bipartite states.

We now define the key convex cone for this section, that is, the cone of positive semidefinite matrices that have diagonal elements with some fixed proportions. Note that the cone of positive-semidefinite matrices with all diagonal elements equal has been well-studied in the literature, under the name of complex \emph{correlation matrices} \cite{Christensen1979corr}. In the case of real matrices, this cone is sometimes referred to as the \emph{elliptope}. Let us define the linear space of matrices for which the diagonal of the matrix is proportional to some fixed non-negative vector. 
\begin{definition}
Let $x \in \mathbb{R}^n_+$ and, 
    \begin{equation}
        \mathcal{M}_n[x] := \{X | \, \exists \lambda, \, \,  \forall i, \,  \,  X_{ii} = \lambda x_i\}.
    \end{equation}   
\end{definition}

Let $e$ be the vector with all elements $1$. 

\begin{definition}[Correlation matrices]
    $\mathcal{M}^+_n[e]= \mathcal{M}_n^+ \, \cap \,  \mathcal{M}_n[e]$ is the cone of PSD matrices with the diagonal elements equal. 
\end{definition}

For any vector $q \in \mathbb{C}^n$, we can embed the vector into a diagonal matrix, $D_q := \operatorname{diag}(q_1, q_2 \ldots q_n)$. We proceed with the following definition. 

\begin{definition}[Scaled correlation matrices]
    For any vector $x \geq 0$, the scaled correlation matrices,  
    $$\mathcal{M}^+_n[x] := D_{\sqrt{x}} \, \mathcal{M}^+_n[e] \, D_{\sqrt{x}}$$ is the convex cone of positive semidefinite matrices with the diagonal $\propto x$. 
\end{definition}

Moreover, the ranks of the extremal matrices in the cone $\mathcal{M}^+_n$ have been completely characterized. It is known that any extremal matrix has a rank-$r$ if and only if $r^2 \leq d$. \cite{Grone1990corr} \cite{Li1994corr}. Hence, we want to point out that the maximum ranks that are achieved by such extremal matrices are $\lfloor \sqrt{n} \rfloor$. Also, it follows that, 

\begin{theorem}[\cite{Li1994corr}]
    For all $n \geq 4$, we have $$\operatorname{R}_1[\mathcal{M}^+_n[e]] \subsetneq \mathcal{M}^+_n[e].$$ 
\end{theorem}

This means that the cone of such matrices is not rank-1 generated for any $n \geq 4$. The same holds for the scaled correlation matrices $\mathcal{M}^+_n[x]$. We show that these sets can also be connected to detecting PPT entanglement in the face $\mathcal{F}_{\ket{\phi}}.$
Let us define $K := \sum_{i=1}^n\ket{i}\bra{ii}.$ Let $\rho$ be a bipartite PPT state in the face $\mathcal{F}_{\ket{\phi}}$. Then, we can see that, 
\begin{equation}
\label{eq:diagonal-fixed}
        Z^\Gamma_K(\rho)_{ii} = |\phi_i|^2 .
\end{equation} 
This implies that $Z^\Gamma_K(\rho) \in \mathcal{M}^+_n[\phi \odot \phi].$ Then, the following proposition follows from the application of \cref{thm:rank-1-generated-ent-transpose}.
\begin{proposition}
\label{thm:correlation-rank-1}
     If the state $\rho$ is separable in $\mathcal{F}_{\ket{\phi}}$, then the matrix, $$Z^\Gamma_K(\rho) \in \operatorname{R}_1[\mathcal{M}^+_n[\phi \odot \phi]]$$
\end{proposition}

Therefore, we find that if such a state in $\mathcal{F}_{\phi}$ is separable, the necessary condition on the matrix $Z^\Gamma_K(\rho)$ is that it has to be a rank-1 generated element of the $\mathcal{M}^+_n[\phi \odot \phi]]$. This also allows us to construct new PPT entangled states. We first provide an example of a witness to detect the violation of the rank-1 generated property in the first non-trivial case of $n=4$, which we borrow from the \cite{jarre2020set}, which outlines a general strategy to witness the separation between correlation matrices and their rank-1 generated version. To this end, note that for complex correlation matrices, i.e, $\mathcal{M}^+_n[e],$ we also have the following characterization, 

\begin{equation*}
    \operatorname{R}_1[\mathcal{M}_n^+[e]]  = \operatorname{cone}\{\ketbra{z}{z} \, \mid \, z \in \mathbb{T}^n\, \}.
\end{equation*} Consider the following $4 \times 4$ matrix, 

\[
W  =\begin{bmatrix}
0 & -i & i & 1 \\
i & 0 & -i & 1 \\
-i & i & 0 & 1 \\
1 & 1 & 1 & 0 & 
\end{bmatrix}.
\] Then, it is shown \cite{jarre2020set} that $\braket{z|W|z} \leq 6$ for all $z \in \mathbb{T}^4$ and hence we have $\operatorname{Tr}(W Z) \leq 6 \operatorname{Tr(Z)}$ for all $Z \in \operatorname{R}_1[\mathcal{M}_4^+[e]]$. Let us consider the following parametrized matrix, 

\[
H(x)  =\begin{bmatrix}
1 & -ix & ix & x \\
ix & 1 & -ix & x \\
-ix & ix & 1 & x\\
x & x & x & 1 & 
\end{bmatrix}.
\]

It is easy to verify by direct computation of its eigenvalues that $H(x) \geq 0 \iff |x| \leq 1/\sqrt{3}$ and also $\operatorname{Tr}(H(x)W) = 12x \leq 6 \iff x \leq 1/2$. Hence, for the parameter range $1/2 < x \leq 1/\sqrt{3}$, the matrix $H(x)$ is not a rank-1 generated correlation matrix.  Note that, for $1/2 < x < 1/\sqrt{3},$ the matrix also has full rank $4$, and is not extremal. 

\begin{theorem}
\label{thm:extremal-correlation}
    The matrix $H(x)$ is rank-$2$ and extremal in $\mathcal{M}^+_4[e]$ for $x = \frac{1}{\sqrt{3}}.$ 
\end{theorem}
\begin{proof}
    The rank-2 property can be immediately verified by the computation of eigenvalues. For extremality, we use the results of \cite[Theorem 1]{Li1994corr}. Consider the matrix $X$ whose columns are the eigenvectors of the matrix. 
    We need to show that the following holds. 
    \begin{equation}
    \label{eqn:span}
        \operatorname{span}\{x_i x_i^* \mid x_i \in X^*\} = \mathcal{M}^{\text{sa}}_2.
    \end{equation}
    Here the matrix, $X^* = \begin{bmatrix}
         1 & 1 & 1 & -\sqrt{3} \\ 
         1 & \omega^2 & \omega & 0
    \end{bmatrix}$ where $\omega := e^{i 2 \pi/3}$. 
    Then it is easy to verify that the \cref{eqn:span} is satisfied for $x_1 = (1,1), x_2 = (1, \omega^2), x_3 = (1, \omega), x_4 = (-\sqrt{3},0)$. 
\end{proof}

Finally, we provide examples of PPT entangled states using extremal matrices with rank $\geq 2.$
We consider the following parametric family of quantum states ($\forall i, A_{ii} = 1$), \begin{equation}\label{eq:corr-states}\ketbra{\omega}{\omega} + \sum^4_{i \neq j} A_{ij} \ketbra{ij}{ij} + \sum^4_{i \neq j} H(x)_{ij} \ketbra{ij}{ji}\end{equation} where $\omega = \sum_i \ket{ii}$ is the maximally entangled Bell state (un-normalized). These states belong to the face $\mathcal{F}_{\ket{e}}.$ The PPT property can be verified readily from \cref{thm:ppt-ldoi} by having $\forall i,j, A_{ij} A_{ji} \geq 1$. Moreover, it is easy to show that by choosing $K = \sum^4_{i=1}\ketbra{i}{ii}$ we have $Z^\Gamma_K(\rho) = H(x).$ Since this matrix is not a rank-1 generated correlation matrix for $1/2 < x \leq 1/\sqrt{3}$, the state is PPT entangled by \cref{thm:correlation-rank-1}. Moreover, for $x = \frac{1}{\sqrt{3}}$ this is also a PPT entangled edge-state by \cref{thm:extremal-correlation}. 

Note that the states considered in the equation \cref{eq:corr-states} are exactly the LDOI states with the triple $(A, \mathbb{J}_4, H(x)).$ We have the following theorem from \cite{singh2021diagonal}, which shows our necessary condition is also sufficient in the case of $A = \mathbb{J}.$  

\begin{theorem}[\cite{singh2021diagonal}]
    The LDOI state in $\mathcal{M}_d \otimes \mathcal{M}_d$ with the triple $(\mathbb{J}_d, \mathbb{J}_d, X)$ is separable if and only if $X \in \operatorname{R}_1[\mathcal{M}_d^+[e]].$
\end{theorem}

We also show that two commonly used criteria for detecting entangled states fail to do so. 
\underline{CCNR criterion}: From \cref{thm:ccnr-ldoi}, it follows that the states we consider in \cref{eq:corr-states} satisfy the CCNR criterion if and only if 
$$\vert \vert A \vert \vert_1 - \vert \vert A \vert \vert_{\operatorname{Tr}} \geq 12.$$

For the case $A = \mathbb{J}_4$ we have $\vert \vert A \vert \vert_1 = 16,$ and $\vert \vert A \vert \vert_1 = 4,$ so the state satisfies the CCNR criterion.  Moreover, by choosing $A = \mathbb{I}_4 + t (\mathbb{J}_4 -\mathbb{I}_4),$ we have that the state satisfies the realignment criterion if $t \geq 1,$ which is also ensured by the PPT criterion. This shows that the rank-based criterion is more powerful than CCNR (and also PPT) for some classes of states.
\end{example} 

\subsubsection{Bose-symmetric states}\label{sec:bose-symmetric}
In this section, we analyze another class of quantum states for which our main result,  \cref{thm:rank-1-generated-ent-transpose}, provides a method of detecting entanglement. We consider the states for which the range lies in the symmetric subspace, which we denote by $\mathbb{C}^d \vee \mathbb{C}^d, $ $\vee$ being the symmetric product. Clearly, this forms a face of the bipartite PSD cone, due to the restriction of its range (see \cref{thm:psd-faces}). 
Let $F$ denote the flip operator, $$F = \sum^d_{i,j=1} \ketbra{ij}{ji}$$ and $\Pi := \frac{\mathbb{I}+ F}{2}$ denote the projection into the symmetric subspace.  Then, for a bipartite PSD matrix, the following are equivalent, 

\begin{itemize}
    \item $\operatorname{ran}(\rho) \subseteq \mathbb{C}^d \vee \mathbb{C}^d,$
    \item  $\Pi \rho \Pi = \rho,$
    \item  $\rho F = F \rho =  \rho.$
\end{itemize}

Interestingly, for all bosonic states, all the usual entanglement criteria, like PPT, realignment, and the covariance criterion, are equivalent \cite{toth2009entanglement}, hence, any PPT entangled bosonic state cannot be detected by such criteria. Let us call the bosonic PSD matrices as $\mathcal{F}_{\mathsf{\vee}}$. Let us consider the convex cone of doubly non-negative matrices, 
$$\mathsf{DNN}_n :=  \mathcal{M}^+_n \cap \mathcal{M}_n (\mathbb{R}_+).$$

Then, we have the following proposition.
\begin{proposition}
\label{prop:bose-symm-dnn}
    Assume $\rho \in \mathcal{F}_{\mathsf{\vee}}$ is PPT. Let $\{\ket{i}\}^n_{i=1}$ denote the orthonormal set of vectors of size $n$. Let $K = \sum^n_{i=1}\ket{i}\bra{i, \overbar{i}}.$ Then, the following is true,
    $$Z^\Gamma_K(\rho) \in \mathsf{DNN}_n.$$
\end{proposition}
\begin{proof}
Let $\rho$ be a PSD state in $\mathcal{F}_{\mathsf{\vee}}.$ Moreover, we have,  
\begin{align*}
Z^\Gamma_K(\rho) 
&= \sum_{i=1}^n \sum_{j=1}^n 
    \braket{i, \overbar{i} \vert \rho^\Gamma \vert j, \overbar{j}} \, \ket{i}\bra{j} \\
&= \sum_{i=1}^n \sum_{j=1}^n 
    \braket{i, j \vert \rho \vert j, i} \, \ket{i}\bra{j} \\
&= \sum_{i=1}^n \sum_{j=1}^n 
    \braket{i, j \vert \rho F \vert i, j} \, \ket{i}\bra{j}\\
&= \sum_{i=1}^n \sum_{j=1}^n 
    \braket{i, j \vert \rho \vert i, j} \, \ket{i}\bra{j}
\end{align*}
Since $\rho$ is PSD, we have $\braket{i, j \vert \rho \vert i, j} \geq 0 \implies Z^\Gamma_K(\rho) \in \mathcal{M}_d(\mathbb{R}_+).$ Moreover, this matrix is PSD by \cref{thm:rank-1-generated-ent-transpose}.
\end{proof}

Using our result in \cref{thm:rank-1-generated-ent}, we can provide another necessary criterion for the separability of bosonic states. 
\begin{corollary}
\label{thm:dnn-rank-1}
    Assume $\rho$ is a separable state in $\mathcal{F}_\vee$. Then, by choosing $K := \sum^n_{i=1}\ketbra{i}{i, \overbar{i}}$, we have $$Z^\Gamma_K(\rho) \in \operatorname{R}_1[\mathsf{DNN}_n].$$
\end{corollary}

\begin{table*}[t]
\centering
\begin{tabular}{r|l|c}

\rowcolor{pink}
Face of PSD &  PPT ${\implies}$ &  SEP ${\implies}$\\ 
\midrule

Sparse States (\ref{sec:psd-graphs}) & Sparse PSD matrices (\ref{thm:sparse-PSD}) & Rank-1 gen. $G$-sparse PSD matrices (\ref{cor:sparse-psd-rank-1})\\ 

\midrule

Restricted Rank-1 States (\ref{sec:corr-matrices}) & Scaled correlation matrices (\ref{eq:diagonal-fixed}) & Rank-1 gen. scaled correlation matrices (\ref{thm:correlation-rank-1})\\ 

\midrule

Bosonic states (\ref{sec:bose-symmetric}) & DNN matrices (\ref{prop:bose-symm-dnn}) & Rank-1 gen. DNN matrices (\ref{thm:dnn-rank-1})\\ 

\end{tabular}
\caption{Summary of the considered face of the state and the mappings of the PPT states on the face. The separable states are mapped to the corresponding rank-1 generated convex cone.}
\label{tab:mytable}
\end{table*} 
The ranks of the extremal rays of the doubly non-negative cone have been scrutinized in \cite{cristaextremednn}.  
\begin{theorem}[Ranks of extremal DNN matrices]
\label{thm:ranks-of-dnn}
    Let $A \neq 0 \in \operatorname{ext} \mathsf{DNN}_n.$  Then, we have, 
    \begin{itemize}
        \item If $n = 2$, $\operatorname{rk}(A) \leq 1.$
        \item If $n > 2$ is even, $\operatorname{rk}(A) \leq n-3.$
        \item If $n > 2$ is odd, $\operatorname{rk}(A) \leq n-2.$
    \end{itemize}
    Moreover, we have elements $A \in \operatorname{ext} \mathsf{DNN}_n$ for all ranks satisfying the conditions above. 
\end{theorem}

Moreover, the rank-1 generated version of this cone is the well-known cone of completely positive matrices.
$$\operatorname{R}_1[\mathsf{DNN}_n] = \{\ketbra{v}{v} \mid v \in \mathbb{R}_+^d\} := \mathsf{CP}_n.$$  \cref{thm:ranks-of-dnn} implies that $\operatorname{R}_1[\mathsf{DNN}_n] = \mathsf{DNN}_n$ if and only if $n \leq 4.$ For more details on these cones and examples of doubly non-negative matrices that are not rank-1 generated (i.e, not CP), we refer you to the book \cite{shaked2021copositive}. The extremal DNN matrices also allow us to construct PPT entangled edge states (that are also bosonic). In general, it is computationally NP-hard to decide whether a $\mathsf{DNN}_n$ matrix belongs to $\mathsf{CP}_n$. Lots of results are known about these matrices in special cases, like matrices having sparsity patterns corresponding to triangle-free graphs \cite{shaked2021copositive}.  

Let us also briefly introduce the witnesses for detecting states that are not rank-1 generated. These are well-known under the name of copositive matrices, denoted as $\mathsf{COP}_n,$ that is, exactly the set of all matrices satisfying, 

$$\forall v \in \mathbb{R}^n_+ \quad \braket{v \vert X \vert v} \geq 0.$$

It is clear from the previous definition that if a matrix is the sum of a non-negative matrix and a PSD matrix, then it is copositive. Such matrices are called $\mathsf{SPN}_n$ matrices. This convex cone is the conic-dual to $\mathsf{DNN}_n$ matrices \cite{shaked2021copositive}, and hence if $X \in \mathsf{SPN}_n$ then $$\forall Y \in \mathsf{DNN}_n \quad \operatorname{Tr}(XY) \geq 0.$$  

To detect non-rank-1 generated matrices, we can construct a witness that is copositive but not SPN. The paradigmatic example of such a matrix is the \emph{Horn's matrix}.
\[
H =
\begin{pmatrix}
1 & -1 & 1 & 1 & -1 \\
-1 & 1 & -1 & 1 & 1 \\
1 & -1 & 1 & -1 & 1 \\
1 & 1 & -1 & 1 & -1 \\
-1 & 1 & 1 & -1 & 1
\end{pmatrix}
\]

To construct an example of a $\mathsf{DNN}_5$ matrix that is not rank-1 generated, we borrow the example from \cite{gulati2025entanglement}. 

\[
A =
\begin{pmatrix}
2\cos{\tfrac{\pi}{5}} & 1 & 0 & 0 & 1 \\
1 & 2\cos{\tfrac{\pi}{5}} & 1 & 0 & 0 \\
0 & 1 & 2\cos{\tfrac{\pi}{5}} & 1 & 0 \\
0 & 0 & 1 & 2\cos{\tfrac{\pi}{5}} & 1 \\
1 & 0 & 0 & 1 & 2\cos{\tfrac{\pi}{5}}
\end{pmatrix}.
\]

By computing $\operatorname{Tr}(AH) = 10 (\cos \pi/5-1) < 0,$ we can verify that the matrix $A$ is not rank-1 generated in the doubly non-negative matrix cone.
\begin{example}
Let us present an example of a family of PPT bosonic edge states that we can detect using the entanglement criterion in \cref{thm:dnn-rank-1}.
$$\sum_{i,j}A_{ij}\ketbra{ij}{ij} + \sum_{i \neq j}Y_{ij}\ketbra{ii}{jj} + \sum_{i \neq j}A_{ij}\ketbra{ij}{ji}$$ for any $A \in \operatorname{ext} \mathsf{DNN}_d$ such that $\operatorname{rk}(A) \geq 2$,  $Y \geq 0$ and $|A_{ij}| \geq Y_{ij}$. We can obtain these states as soon as $d \geq 5$, for constructions see \cite{cristaextremednn}. It can be verified to be PPT entangled, bosonic, and an edge state by choosing $K:= \sum^d_{i=1} \ket{i}\bra{i, i}.$ and applying the previous results (\cref{thm:dnn-rank-1}, \cref{thm:rank-1-generated-ent-transpose}).
\end{example}

We want to point out that for $Y_{ij} = 0, $ the states considered above are called \emph{mixtures of Dicke states} \cite{tura2018separability, singh2021diagonal}, which are separable if and only if $A \in \operatorname{R}_1[\mathsf{DNN}_n].$

\section{Discussion and Outlook}
In this article, we presented a rank-based criterion to detect PPT entanglement in mixed bipartite quantum states by mapping PPT states to cones that are not rank-1 generated. This allows us to use a variety of results and examples of matrix analysis, optimization, and graph theory to construct new methods and examples of PPT entangled states.

Our work definitely opens avenues for future research. The entanglement-theoretic properties of the states constructed by the above methods should be scrutinized (for example, Schmidt numbers and entanglement measures). It remains to be seen whether PPT entangled states can provide meaningful advantages over separable (unentangled) states in quantum protocols, some evidence for which has been recently obtained in \cite{tavakoli2025unlimited, toth2018quantum, pal2021bound} in the area of quantum metrology and non-local games. We hope that the methods used in this paper and examples of PPT entangled states are useful test states for such advantages. It will also be useful to analyse whether the witnesses that detect matrices that are not rank-1 generated in all cases considered in this paper can be used to construct positive indecomposable maps, an important topic of research in operator algebras. This was recently accomplished for copositive matrices \cite{gulati2025positive}, the witnesses discussed in \cref{sec:bose-symmetric}.

Finally, the faces of the PSD cone we consider can be mapped to known convex cones that are not rank-1 generated. In general, it would be useful to understand how this machinery can be expanded to other faces of the PSD cone and appropriately chosen mappings, to construct new interesting examples of PPT entangled and PPT edge-states. It might be useful to first understand this question in low-dimensional quantum systems, i.e, in the case of two qutrits, or qubit-ququarts. 

\bigskip

\section{Acknowledgements}
The author thanks Ion Nechita, Sang-Jun Park, and Satvik Singh for insightful discussions on this topic. The author received support from the University Research School EUR-MINT (State support managed by the National Research Agency for Future Investments program bearing the reference ANR-18-EURE-0023) and ANR project \href{https://esquisses.math.cnrs.fr/}{ESQuisses}, grant number ANR-20-CE47-0014-01.

\end{document}